\documentclass[amsmath,amssymb,footinbib,10pt,superscriptaddress]{article}
\pdfoutput=1

\usepackage[utf8]{inputenc}
\usepackage[T1]{fontenc}
\usepackage{amsmath}
\usepackage{amsthm}
\usepackage[a4paper, left=2.5cm, right=2.5cm, top=2.5cm, bottom=2.5cm]{geometry} 
\usepackage{graphicx}
\usepackage{amsfonts}
\usepackage{dsfont}
\usepackage{color}
\usepackage{mathtools}
\usepackage{comment}
\usepackage{float}
\usepackage{listings}
\usepackage{indentfirst}
\usepackage{hyperref}

\newtheorem{lemat}{Lemma}
\newtheorem{tw}{Theorem}
\newtheorem{definition}{Definition}
\newtheorem{corollary}{Corollary}
\newtheorem{property}{Property}

\title{Analysis of Multiple Overlapping Paths algorithms for Secure Key Exchange in Large-Scale Quantum Networks
}

\author{Mateusz St\k{e}pniak, Jakub Mielczarek\\
Institute of Theoretical Physics, Jagiellonian University, \\
{\L}ojasiewicza 11, 30-348 Cracow, Poland}

\date{\today}

\begin{document}

\maketitle

\begin{abstract} 
Quantum networks open the way to an unprecedented level of communication security. 
However, due to physical limitations on the distances of quantum links, 
current implementations of quantum networks are unavoidably equipped with 
trusted nodes. As a consequence, the quantum key distribution can be performed only 
on the links. Due to this, some new authentication and key exchange schemes 
must be considered to fully benefit from the unconditional security of links. 
One such approach uses Multiple Non-Overlapping Paths (MNOPs) for key exchange 
to mitigate the risk of an attack on a trusted node. The scope of the article 
is to perform a security analysis of this scheme for the case of both 
uncorrelated attacks and correlated attacks with finite resources. Furthermore, 
our analysis is extended to the case of Multiple Overlapping Paths (MOPs). 
We prove that introducing overlapping paths allows one to increase the security 
of the protocol, compared to the non-overlapping case with the same number 
of additional links added. This result may find application in optimising 
architectures of large-scale (hybrid) quantum networks. 
\end{abstract}

\section{Introduction}

Public-key cryptography discovered in the 1970s of the last century provided 
a long-sought solution to the problem of secret key exchange \cite{DiffieHellman}. 
The ingenious breakthrough allowed for the promotion of ideas such as secure 
communication on the Internet. The security of public-key cryptography is based 
on the high computational complexity of the problem used, such as factorisation 
of composite numbers. 

When public-key cryptography has already been implemented on a global scale, it has 
been realised that the complexity of the problems used can be reduced if 
quantum computing resources are used \cite{Shor}. This raised concerns about the 
security of public-key cryptography with respect to hypothetical quantum attacks. 
However, the solution to the problem already existed (theoretically) and relied 
not on computational complexity-based security but on information-theoretic security 
(ITS). The solution is the quantum key distribution (QKD) \cite{Ekert91}.

While extremely appealing from the theoretical viewpoint, technical difficulties 
precluded the wide implementation of QKD for decades, since its theoretical introduction 
in the 1980s of the last century. However, the situation has improved significantly 
in the last few years and QKD solutions are blossoming. However, the remaining significant 
obstacle is the distance on which the QKD can be performed. This is due 
to the suppression of photons in the optical medium. Therefore, ground-based optical 
fibre links allow for practical QKD (sufficiently high key exchange rate) over 
distances not longer than approximately 100 km \cite{range}. A possible, but costly 
and challenging, solution to the problem is to utilise space. Due to the much weaker
suppression of photons in air and cosmic vacuum, QKD can be performed at much longer 
distances \cite{range}.  
   
The potential solution to the problem of the constraint on the distance
at which the QKD can be performed on the ground is given by quantum 
repeaters. However, the technology is not mature enough to be implemented     
in the present realisations of the QKD solutions. As a consequence, in 
the current implementations of QKD networks (quantum networks), 
classical trusted nodes must be used. The QKD networks are, consequently, 
hybrid networks with quantum links and classical nodes. 
Examples of experimental realisations of such networks are: 
Tokio network \cite{tokio}, Beijing-Shanghai network \cite{Beijing-Shanghai}, 
Madrid network \cite{madrit}.

However, the existence of the classical nodes raises security concerns.
Although the QKD link has been shown to be ITS-safe, nodes can become 
a source of information leakage. The purpose of this article is to present 
a scheme that will significantly improve the security of hybrid 
QKD networks. 

Notably, at this point, QKD algorithms are also potentially vulnerable  
to the man-in-the-middle (MITM) attacks. Usually, this problem is resolved 
by applying the authentication of the classical nodes. Since the idea 
of employing QKD is to eliminate non-ITS protocols, the authentication 
of QKD nodes must be performed with the use of the Wegman-Carter protocol,
which has been proven to be of the ITS class. More information about QKD 
authentication can be found in Ref. \cite{autentykacja}. Although the 
authentication problem can be successfully resolved in this way, this does 
not concern end-to-end encryption (E2EE). 

From a theoretical point of view, using QKD combined with a one-time pad (OTP) 
guarantees end-to-end ITS communication. However, with actual limitations on 
the key exchange rate, this approach is too slow to be used instead of classical
communication. An alternative to OTP is to use weaker (non-ITS) symmetric 
cryptography algorithms that are generally resistant to quantum attacks 
\cite{post-quantum-crypto} and this seems to be the way the QKD network could 
be used. A completely different direction for preparing for quantum attacks, 
which does not rely on QKD network, is post-quantum public-key cryptography, 
which has been intensively developed in recent years \cite{post-quantum-crypto2}. 

In this work, we consider the potential vulnerability in the quantum network 
based on trusted nodes with multiple paths \cite{multipath}, we construct two 
models of possible attacks, the uncorrelated attack, in which each node has 
a certain probability of becoming compromised, and the correlated attack, in 
which the opponent owns certain resources that could be used to compromise 
the security of certain nodes. Furthermore, based on \cite{Zhou2019SecurityAA}, 
we extend the QKD multiple-path distribution protocol, assuming that the paths 
can overlap. Our analysis shows that to improve security, it is (under certain 
conditions) more optimal to add interlinks between disjoint paths, instead of 
adding a new path. A similar concept using overlapping paths has been 
presented in \cite{podobnypomysl}, where security is improved by introducing
complete subgraphs (``cities''). While preparing this article, another work 
that addresses security issue of overlapping multiple paths has appeared
\cite{securityrosjanie}.

\section{The multiple paths protocol}
\label{sec:protocol}

One of the central concepts behind the design of a telecommunication network 
is redundancy. To marginalise the probability of a lack of connectivity between 
two nodes, there must be at least two alternative paths that connect arbitrary 
two nodes. Here, we assume that the same must concern QKD networks, in particular 
the hybrid QKD networks under consideration. This approach has been proposed 
in \cite{multipath} and recently explored in \cite{solomons2021scalable}, in 
which optimal key flooding is considered.  

Now to help the reader better understand further material, we introduce problems 
concerning security and concepts of multipath protocol. Consider a simple 
uncorrelated attack scheme, in which the possibility of a successful attack on 
any node is $p$ and $n$ is the total number of intermediate nodes on the path 
(excluding communicating nodes/parties). The probability that at least one node 
has been successfully attacked is $P_1=1-(1-p)^n \approx np$, where the approximation 
is valid for $p \ll \frac{1}{n}$. So, roughly the probability of an attack 
on the network grows with the number of nodes. If the network is a hybrid QKD 
network, this would be equivalent to leaking a secret key exchanged via the 
attacked node. 

Following this simple model, we find the probability that at least
two nodes have been attacked: $P_2 \approx  (np)^2$.  Therefore, under 
the condition $p \ll \frac{1}{n}$, the probability of a successful
simultaneous attack on at least two nodes is quadratically lower than
in the previous case. 

Following the above observations, let us consider a scenario in which 
a secret key $K$ is composed of two parts $K_1$ and $K_2$ of equal 
length. For example, if the key $K$ is devoted to be applied in the 
symmetric AES-256 algorithm, both parts $K_1$ and $K_2$ are 256 bits 
long. Now, we require that knowing one of the keys gives us zero 
knowledge of the key $K$. This requirement can be easily satisfied by 
the One-Time Pad (OTP) applied to the two parts $K_1$ and $K_2$, so that: 
\begin{equation}
K = (K_1,K_2) = K_1 \oplus  K_2,
\label{OTPKK}
\end{equation}
where $\oplus$ is the XOR operation (addition modulo two).  

Now, let us suppose that $K_1$ and $K_2$ are two bit 
strings that are distributed using two different paths 
in the QKD network. The two paths connect two parties 
(nodes $A$ and $B$), between which the key $K$ is 
exchanged (see Fig. \ref{SSQKDPlot1}). For every two 
adjoint nodes on the network $(i,j)$ a secret key $K_{ij}=K_{ji}$ 
is established via QKD. Then, in a \emph{hop-by-hop} 
approach, if node $i$ wants to send a secret message 
$M$ to the node $j$, the OTP encryption is used by 
evaluating the cipher $C=M\oplus K_{12}$. The classical 
ciphertext $C$ is transmitted to the subsequent node 
via the classical (untrusted) channel. Then, by evaluating  
$M=C\oplus K_{12}$, the ciphertext is decrypted at the 
node $j$.

\begin{figure}[ht!]
\centering
\includegraphics[width=5cm,angle=0]{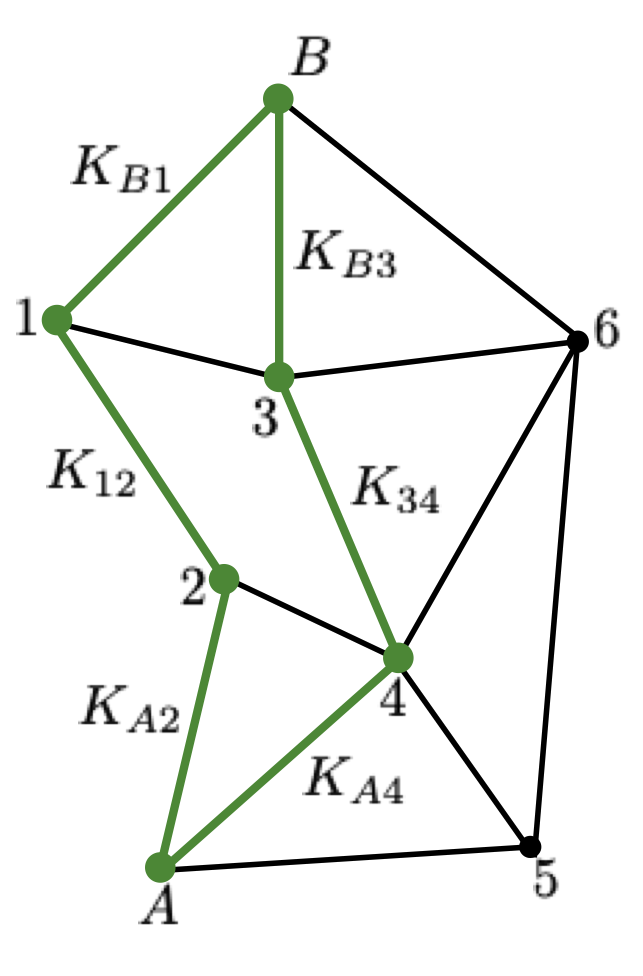}
\caption{Pictorial representation of the exchange 
of a secret key $K=(K_1,K_2)$ shared in two paths
connecting the communicating nodes $A$ (transmitter) 
and $B$ (receiver). Here, the part $K_1$ is exchanged
via the nodes 1 and 2, while the part $K_2$ is exchanged
via the nodes 3 and 4. In the first step the, the 
partial keys $K_1$ and $K_2$ are sent from the 
node $A$ to the nodes $2$ and $4$ respectively. 
In both cases the OTP encryption with the QKD 
keys ($K_{A2}$ and $K_{A4}$) are used. The procedure 
is then continued, via the nodes $1$ and $3$, to 
the end node $B$.}
\label{SSQKDPlot1}
\end{figure}

The protocol introduced in this section quadratically 
improves the security of secret key exchange in a hybrid 
QKD network. However, this is achieved by the cost of 
doubling the number of keys exchanged. Therefore, we 
reduce the performance of the system by a factor of 
two. However, the quadratic improvement by the linear 
cost seems to be a beneficial solution.    

The decomposition of the secret key into two parts is 
an example of a secret sharing. It is worth emphasising 
that the idea can be generalised by splitting the secret 
key into three or more constituents. In this case, 
Eq. \ref{OTPKK} generalises to:
\begin{equation}
K = (K_1,K_2,\dots,K_N) = K_1 \oplus K_2\oplus \dots \oplus K_N,
\label{OTPKN}
\end{equation}
where $N$ is the number of different paths. The realisation 
of the case with $N>2$ allows for a further reduction of 
the probability of the attack to $\sim (np)^N$. However, 
this is due to the cost of both much complex topology of 
the QKD network and its lower performance. In the next 
sections, we perform a detailed analysis of two types of 
attack on trusted nodes and in Sec. \ref{sec:multipathscheme} 
we present the concept in which key exchange can be done 
with paths that cross each other along with the discussion
of performance of this new protocol.

\section{Security considerations on Multiple 
Non-Overlapping Paths}\label{sec:attacks}

We consider two models of possible attack, first for an 
uncorrelated attack that simulates leaking and publishing 
information (secret key) steaming from random failure of 
certain nodes, and second for a correlated attack where 
the party is assumed to have certain resources that can 
be used to take over some nodes and gain information. 
We model the QKD network by an unweighted graph, where 
two communicating nodes can always be connected by a 
certain number of disjoint paths. This is justified as 
a desired property of a real telecommunication network 
\cite{Elliott_2002}.

\subsection{Uncorrelated attacks}
\label{uncorrelated-attack}

\subsubsection{General formulation}

{\it Given a graph with two distinguished nodes named $A$ and $B$ 
(Alice and Bob), the $i$-th node is marked with a certain number 
$p_i$ representing the probability that during the protocol the 
$i$-th node will be hacked and publicly reveal secret key 
($p_i=0$ represent complete trust while $p_i=1$ represent 
fully corrupted node). $A$ and $B$ can be connected with some 
disjoint paths, and to compromise security of protocol at least 
one node on each path must become untrusted. The task is to find 
a system of disjoint paths between $A$ and $B$ that minimise the 
probability of hacking communication.}

For a predefined system, calculating the probability is an easy 
task. Let $\mathcal{R}$ be the family of all paths in the solution 
(where the path is considered to be a set of intermediate nodes,
\emph{i.e.} excluding the $A$ and $B$ nodes). Then, the probability 
of hacking is given by:
\begin{equation}
   P = \prod_{\mathcal{R}_j \in \mathcal{R}} \left(1-\prod_{i \in \mathcal{R}_j} (1-p_i)\right). 
\end{equation}
Here, each node is identified with its label $i$. 

This task poses an algorithmic challenge, and to the best 
of our knowledge, there is no standard effective method to 
solve this in this form. Instead, we could search for a 
strategy that guarantees a certain threshold of security 
level and is flexible within this limit. As a benefit, this 
may also allow us to adjust the algorithm to network traffic. 
In the next paragraph, we present an approach to the problem 
with respect to the considerations mentioned above.

\subsubsection{Simplified problem}

Here, we assume that each node has the same probability $p$, 
which is small enough so that we can use the approximation 
$1-(1-p)^n \approx np $. The approximation is satisfied
under the assumption that $np \ll 1$ and is further analysed 
in Sec. \ref{sec:error}. According to the model above, we 
can simply take $p :=\max\limits_{i}  (p_i)$. If we consider 
$N :=|\mathcal{R}|$ paths each containing $n_j :=|\mathcal{R}_j|$ 
(intermediate) nodes, then the probability of hacking equals:
\begin{equation}
P=(n_1 p) (n_2 p) ... (n_N p) \leq  p^N \left(\frac{\sum\limits_{1\leq j \leq N} n_j}{N} \right)^N = \left(  \overline{n} p \right)^N,
\end{equation}
where $\overline{n}$ denotes the average path length (in the sense of the number 
of intermediate nodes), and we use a well-known inequality: 
\begin{equation}
\label{ineqality}
    \sqrt[N]{a_1 a_2 ... a_N} \leq \frac{a_1+a_2 +... + a_N}{N}.
\end{equation}
Notice that the equality is satisfied for $a_1=a_2=...=a_N$,
which in the case under consideration corresponds to the paths 
of equal length. In a real QKD network, we may expect that the 
lengths of the paths are similar, and in this case the given 
upper bound could be a good approximation.

\subsubsection{Solution to the simplified problem}

Now we can slightly reformulate the problem so that we 
do not minimise $P$ but $P':=(\overline{n} p)^N$ instead. 
This depends on two factors: the number of paths and the 
average length of these paths. In many analyses, equally 
length paths are considered, and therefore it is always 
optimal to use as many paths as possible, but it turns 
out that it may not be desired for an arbitrary network. 
The analysis of an educational example is provided in \nameref{sec:app_attack_example}.
For fixed $N$ solution which is given and can be obtained 
by the minimum-cost flow algorithm (with unit capacities 
and certain transformation of graph), there also exist other 
simpler algorithms like Suurballe's algorithm (vertex disjoint 
path version) (see Ref. \cite{disjoint_paths}). We may
assume that in realistic case $N$ will not exceed 10. 
Therefore, an efficient algorithm could be obtained by 
checking each possible number of paths separately. 

One last remark is about the practical aspect of the obtained 
solution. If we accept loss in security level (but within 
threshold bound), we can try to add traffic management within 
the algorithm simply by manipulating weights of edges. Undesired 
routes will be less likely to be chosen. However, in this article, 
we do not develop this concept further - it is left for future work.

\subsubsection{Multiple communicating parties}

The problem arises when we have more than two communicating parties, 
and we do not allow a path to share a link (for quantum networks 
the effectiveness of links is the main restriction). This is an 
extensively studied problem called $k$-EDP ($k$-edge-disjoint path 
problem) \cite{EILAMTZOREFF1998113}

\begin{definition}
Consider the following well-known problem, which is called the 
$k$-disjoint paths problem ($k$-DPP). For a given graph 
$G$ and a set of $k$ pairs of terminals in $G$, the objective is 
to find $k$ vertex-disjoint paths connecting given pairs of terminals 
or to conclude that such paths do not exist.
\end{definition}

It is proven that this problem is NP-complete \cite{k-DPP_KARP}. 
Therefore, finding many paths for each of the pairs $k$ is ``at least 
of class NP,'' as this problem can be reduced to $k$-DPP by 
adding an appropriate number of paths between distinguished pairs 
of terminals. However, if $k$ is fixed, polynomial solutions exist 
for $k$-DPP and even for shortest $k$-DPP \cite{lochet2020polynomial,KAWARABAYASHI2012424}, so we can hope 
to search for a solution while dividing the network into clusters. 

\subsection{Correlated attacks with finite resources}

As in the previous section, let us first formulate a general problem:

{\it Consider a graph with a distinguished pair of nodes $(A,B)$. 
We assume that there exists a system of (disjoint) paths connecting 
$A$, $B$ and the adversary knowing this system, having some resources 
which can be used to take control of nodes and extract keys. The 
following assumptions are made:
\begin{itemize}
    \item Adversary has full knowledge about the arrangement of the system.
    \item Hacking one node on the path makes this path untrusted.
    \item Communication is hacked if each path is untrusted.
    \item Probability $p_i$ that the $i$-th node becomes untrusted depends 
    on amount of allotted resources $r_i\geq 0$ and is given with 
    proper continuous function $p_i = f_i(r_i)$ specified for this node. 
    Because $f_i(r_i)$ has an interpretation of probability, 
    it takes values from the range $[0,1]$.
    \item Resources are bounded, that is, $\sum\limits_i r_i \leq R$.
    \item For each node $f_i(0)=0$.
\end{itemize}
We seek tactics (system of paths) that minimise the probability of 
hacking communication.}
\\

\indent As before, solving problem in this form poses a challenge, 
and even for a fixed system, calculating the minimal probability of 
hacking (corresponding to optimal redistribution of resources) is 
difficult due to the continuous character of variables and unknown 
functions. Therefore, again, we need simplification.

\subsubsection{Simplified problem}

We first make the following observation:
\begin{lemat}
Without loss of generality, we can assume that each function 
$f_i$ is not decreasing and the condition $\sum\limits_i r_i = R$ is used.
\end{lemat}

\begin{proof}
We do not need to use all resources, so if it is optimal to use $x_2$ 
resources for a certain node and there exists a $x_1$, such that $x_1 <x_2$ 
and $f(x_1)>f(x_2)$, then it is optimal for the adversary to use $x_1$. 
The adversary will obtain the same result using the alternative function
$f(x)=\max\limits_{y\leq x}(f(y))$. The second part of the lemma is 
straightforward.
\end{proof}

In real communication networks, we can assume that each node does not 
differ much in construction, and thus their characterisation will have 
much in common. At the same time, one shall not allow the adversary to 
easily take control of the node (than the mean $f(R) \ll 1 $), so the 
arguments of the function could be considered small, which allows us 
to expand the function $f$ in series and consider its linear approximation. Alternatively, for a given function $f(x)$, we can construct a new function 
$g(x)$ that is linear up to a certain point, then constant (equal 1) 
and satisfy $f(x) \leq g(x)$.

We summarise this discussion with the following additional assumptions 
for the problem:
\begin{itemize}
    \item The function $f$ is the same for all nodes, \emph{i.e.} 
    $\forall_i f_i=f$.
    \item The function $f$ is not decreasing, and the opponent always uses all available resources.
    \item The function $f$ is expressed in following form:
    \[
    f(x)= 
\begin{dcases}
    \alpha x,&  x\leq \frac{1}{\alpha}\\
    1,              & x\ge \frac{1}{\alpha}
\end{dcases}.
\]
\end{itemize}

With this simplification it turns out, that a sensible analysis can be made. 
We first develop optimal adversary strategy for a single path of length $n$.
To solve this problem we use Lagrange multiplayer method. The function we want 
to maximize (probability of hacking) is: 
\begin{equation}
P_{\text{sp}}(r_1,r_2...,r_n)=1-(1-f(r_1))(1-f(r_2))...(1-f(r_n)),
\end{equation}
with constrain $G(r_1,...,r_n)= \sum_i r_i -R = 0$.

As a result, we see that among the candidates for the global extremum 
(points $(r_1,r_2,...,r_n)$) some of $r_i=0$ and the rest are equal 
to each other. Therefore, the set of extreme values is $\{ P_k | 
k \in {1,...,n} \}$, where $P_k=1-\left(1-\alpha \frac{R}{k}\right)^k$ 
and since this represents the Euler sequence that is decreasing, we 
obtain the global maximum for $k=1$, which corresponds to placing all 
available resources on a single node.
\\

We summarise it in the following theorem:
\begin{tw}
Given a single path, the optimal strategy for an adversary 
is to attack a single node, and the probability of hacking is $\alpha R$. 
\end{tw}

From this we obtain an important conclusion about the situation with many paths.
\begin{corollary}
Given $N$ disjoint paths, the optimal strategy for an adversary is to attack only 
one node on each path.
\end{corollary}

\begin{proof}
Let us assume on the contrary that in optimal strategy for the opponent 
there exists a path on which two or more nodes are attacked. If we 
relocate resources from these nodes to a single node on this path, we 
will have a higher probability of hacking this path and, therefore, 
obtain a better strategy as the probability of hacking system is the
product of probabilities for individual paths.
\end{proof}

If we have $N$ paths and $r_j$ are resources used to hack the $j$-th 
path (at a single node), then the probability of hacking the protocol is:
\begin{equation}
    P(r_1,r_2,...,r_N) = (\alpha r_1)(\alpha r_2) ... (\alpha r_N) \leq  \left(\frac{\alpha R}{N} \right)^{N}.
\end{equation}
Here, we assumed that $R \ll \frac{1}{\alpha}$, and we used Eq. \ref{ineqality}.
 
\subsubsection{Solution}

Choosing the hacking strategy, we can focus on minimising term
$\left(\frac{\alpha R}{N} \right)^{N}$. There is only one variable 
to control: the number of disjoint paths. As $\alpha R$ in this 
approximation shall always be less than $1$ (otherwise, the 
approximation we used fails), we are interested in increasing $N$. 
The maximal number of disjoint paths between pair $A$ and $B$ is 
equal to the minimal size of the vertex cut of that pair, thanks 
to Menger's theorem. 

Through vertex cut is a notion that has no 
unique definition in the literature, we restate it here:
\begin{definition}
\textbf{A-B vertex cut} is a set of vertex that does not 
contain $A-B$ so that after the removal of this set from 
the graph, there is no path between $A$ and $B$. Later in 
the article, we will refer to it as cut, while the default 
vertex $A$ and $B$ will be sender and receiver (Alice and 
Bob). We say that the vertex cut is \textbf{minimal} if 
there is no cut with a smaller order.
\end{definition}
We summarise our conclusions with the following theorem:
\begin{tw}
To improve security against correlated attack for users $A$ and $B$, 
the desired strategy is to increase the order of the $A-B$ minimal 
vertex cut, that is, the number of disjoint paths.
\end{tw}

Finding order of minimal vertex cut is a problem equivalent to (after a simple transformation of graph) solving the \emph{max-flow problem}.

\section{Multiple Overlapping Paths scheme} \label{sec:multipathscheme}

We have performed an analysis of the security of multiple paths
scheme models under the assumption that all paths are disjoint. 
We now present a  Multiple Overlapping Paths scheme (MOPs), 
where additionally to system of disjoint of paths the interpath 
links exist. Such an extension can always be made without loss of 
security, with only slight modification of the well-known 
\emph{hop-by-hop} protocol. Unlike MNOPs, where each intermediate 
node has exactly two links, now it can have more. A similar problem 
was previously analysed in Ref. \cite{podobnypomysl}. However, 
to improve security, the total number of links is increased. 

In this article we follow an alternative idea, increasing security 
with the use of interlinks but without changing the number of links. 
Roughly, we can say that we completely remove one path and use 
its resources (links) to make interconnections between the rest.

\begin{definition}
In the MOPs communication scheme, each node in the network, except 
Alice and Bob, sends the $XOR$ result of all the keys from the neighbour 
connections to Bob via an unencrypted (but authenticated) channel 
(available to Eve). The shared key will be the $XOR$ of all such 
received messages for Bob, and for Alice $XOR$ of all subkeys Alice 
shares with intermediate nodes.
\end{definition}

This idea was originally published in \cite{Zhou2019SecurityAA}, 
where a detailed analysis is included. Here, we recall only the most 
important conclusions. For the case without interlinks between paths, 
this will work as in the classic scheme, with the difference that 
the message is transmitted to Bob instead of to the next node on the 
path, but in \nameref{sec:appNDMS_to_hopbyhop} we show that MOPs can 
be modified so that it mimics the hop-by-hop method, which is important 
for the practical use and efficiency of the network. 

\begin{property}
In MOPs connection is secure only if there exist a path with node 
controlled by adversary. Therefore, to hack a communication between 
$A$ and $B$, adversary must control $A-B$ vertex cut.
\end{property}

There is no reward in using MOPs against uncorrelated attack. It can 
be shown that adding interlinks to a system of disjoint paths does not 
increase the order of minimal $A-B$ vertex cut (for example, because 
the number of nodes connected to Alice does not change and these nodes 
form a vertex cut). However, it can change the number of such cuts and, 
therefore, turns out to be useful against uncorrelated attack. In general, 
to calculate probability that protocol is compromised, one must know 
trustfulness of all nodes and calculate probability that at least one 
of $A-B$ cuts (not necessarily minimal) becomes untrusted. This can be 
really challenging problem in general. We managed to perform analytical 
analysis on certain special type of ``grid-like'' graph that could model 
real QKD network. The model is discussed in the following paragraphs.

\subsection{Intuitive approach}
\label{sec:intuitive_approach}

We begin with an intuitive assertion that adding interlinks instead 
of a new path could perform better than MNOPs. If the probability 
of compromising the security of a node (and therefore leaking the key) 
is sufficiently small, we can take into account only the smallest 
(minimal) $A-B$ vertex cut. The number of such combinations of nodes 
in the classical scheme is 
\begin{equation}
\prod\limits_{1\leq j \leq N}n_j \approx (\overline{n})^N,
\end{equation}
where $n_j$ is the number of intermediate nodes in $j$-th path 
and $N$ is the number of paths. We will improve security by reducing 
the number of $A-B$ vertex cuts in the graph by adding interlinks 
between paths. An example is shown in Fig. \ref{fig:schem-example}. 

\begin{figure}[h]
    \centering
    \includegraphics[scale=0.8]{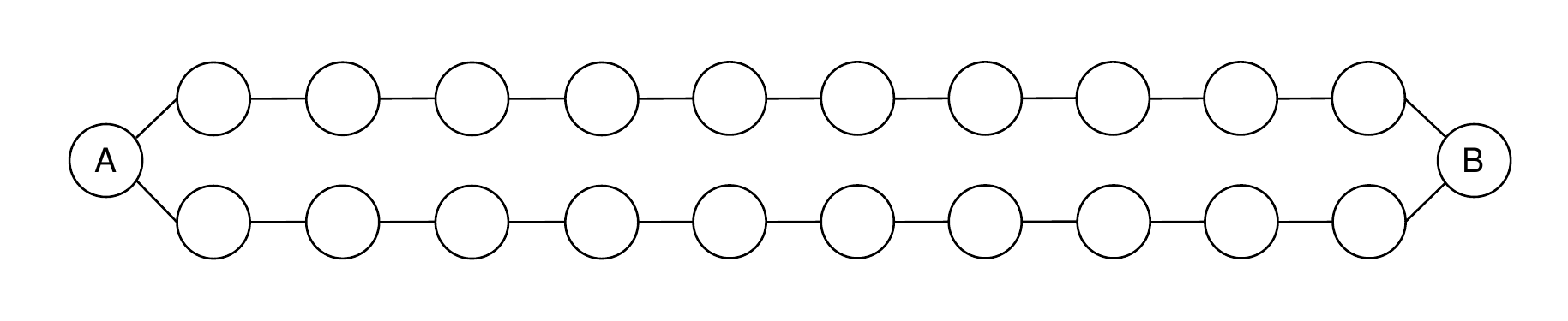}
    \includegraphics[scale=0.80]{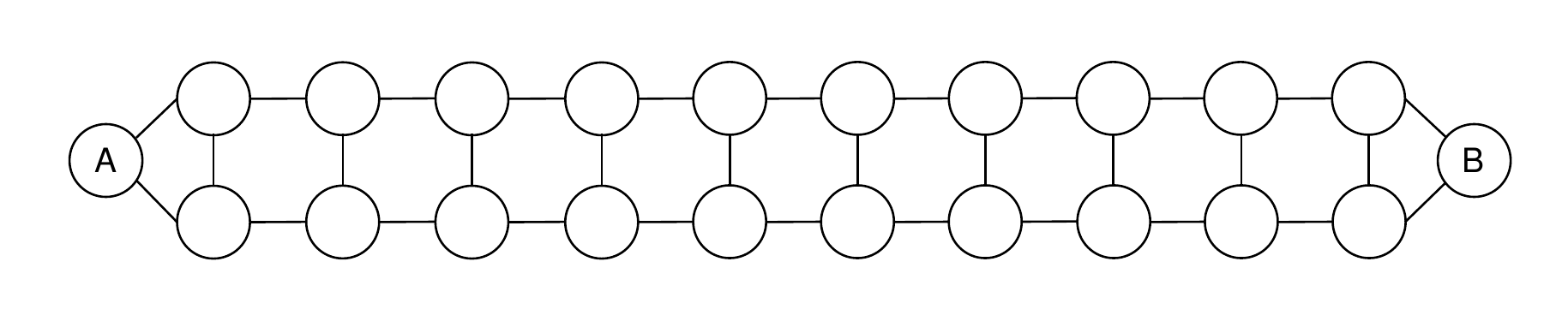}
    \caption{MNOP network (above) with 2 disjoint paths 
    and with additional interlinks (below).}
    \label{fig:schem-example}
\end{figure}

Let us assume that the number of intermediate nodes between $A$ and 
$B$ is $n$ (for every path) and we have initially 2 disjoint paths. 
The number of 2-cuts (cuts of size 2) is $\approx 3n$, for the MNOP 
scheme it would be $n^2$. If $p$ is sufficiently small, we can state that 
only minimal cuts influence the probability of hacking. Then, the 
probability of compromising the security of the MNOP scheme is 
$(np)^2$, while for MOPs: 
\begin{equation}
\underbrace{{2 n \choose{2}}p^2(1-p)^{2n-2}}_{\text{probability of hacking pairs}} \times \underbrace{\frac{3n}{{2 n \choose{2}}}}_{\text{allowed pairs}} 
\approx 3np^2.
\end{equation}
Formula stems from binomial distribution (chance that we control 2 nodes) 
and the probability that these two nodes form the desired hacking cut. 
This can be a significant advantage when $n$ is large. However, we use 
additional links that, in turn, could be used to make another path. 
We now pose the question: ``Can adding interlinks, instead of a new path, 
be a desired strategy?''. It appears that in some cases, especially for 
large $n$, this can be true. Comparing the graph with interlinks and the 
MNOP scheme with one with one additional path, we obtain probability,
respectively $(np)^3$ and $(3np^2)$. The ratio of probabilities is
\begin{equation}
    \eta := \frac{P_{\text{MOP}}}{P_{\text{MNOP}}} \approx \frac{3n p^2}{(n p)^3}
    = \frac{3}{p n^2}.
\end{equation}
If $\eta <1$, it is optimal to use the proposed strategy. But this poses 
a condition on $p$, namely:
\begin{equation}
 p> \frac{3}{n^2}.
\end{equation}
At the same time we have assumed that $p$ is ``sufficiently'' small 
(because we neglected the influence of non-minimal cuts), in fact, 
it must at least meet conditions as in chapter \ref{sec:protocol}, 
\emph{i.e.} $p \ll \frac{1}{n}$. Now, if $n$ is sufficiently large, 
these conditions do not lead to contradiction. In the next chapters, 
we formalize and generalize the above considerations.

\subsection{Adding a single link}
\label{sec:adding_single_link}

After the discussion in the previous section, we could easily notice 
an important property, namely that adding just one link can reduce 
the probability of hacking by a factor of $\approx 1/2$. This approximation 
becomes more accurate as $n$ increases. Consider a MNOPs network with 
two disjoint paths as in Fig. \ref{fig:schem-example}, if we add a 
single intermediate link somewhere in the middle of the network and 
perform similar considerations as in Sec. \ref{sec:intuitive_approach}, 
we obtain that the number of 2-cuts is $\approx  2 \left( \frac{n}{2} \right)^2= \frac{n^2}{2}$. Consequently, for a network with $l$ paths, we can add 
$l-1$ links and reduce the probability of hacking by a factor $\left(\frac{1}{2}\right)^{l-1}$. This can be an important property, 
as by using only a few links, quite a good profit is obtained. 
Unfortunately, this effect does not stack: just from the example 
for two paths presented in Fig. \ref{fig:schem-example} we see that 
adding $n$ interlinks gives as probability reduction by factor 
$\frac{3}{n}$ not $\left(\frac{1}{2}\right)^n$.

\subsection{Formal consideration}\label{subsec:formalconsiderations}

Now, we want to compare two situations, MNOPs scheme with $l+1$ disjoint 
path and strategy presented in Sec. \ref{sec:intuitive_approach} (MOPs) 
where on the behalf of interlinks we remove one path. For the new strategy,
Alice and Bob are connected through $l$ disjoint path each containing $n$
intermediate nodes and each vertex having the same trust level $(1-p)$.  
We introduce the numbering system of vertices: $A$, $B$ for two communicating
vertex, $g_{i j}, \ i \in \{1,...,l\},\ j \in \{1,...,n\}$ for intermediate 
nodes where $i$ and $j$ denote row and column number, respectively.
Adding interlinks can be realized on different strategies. Here, we develop 
a strategy called the \textit{MOPs l-scheme}, one that is probably not optimal 
but provides the possibility of analytical analysis. In the MOPs $l$-scheme 
we connect all vertically adjoined nodes in every ($l-1$)-th column (connecting 
take $l-1$ edges).

Formally, new graph is obtained by adding a set of edges: 
\begin{equation}
\mathcal{E} =\{ g_{ij} \leftrightarrow g_{(i+1)j} |   j \equiv 0 \pmod{l-1}  \; \wedge \; i<l \}.
\end{equation}
It can be easily seen that the number of interlinks created is at most $n$.
The example is presented in Fig. (\ref{fig:l-scheme}).

\begin{figure}[h]
    \hspace{1,5cm}
    \includegraphics[scale=0.7]{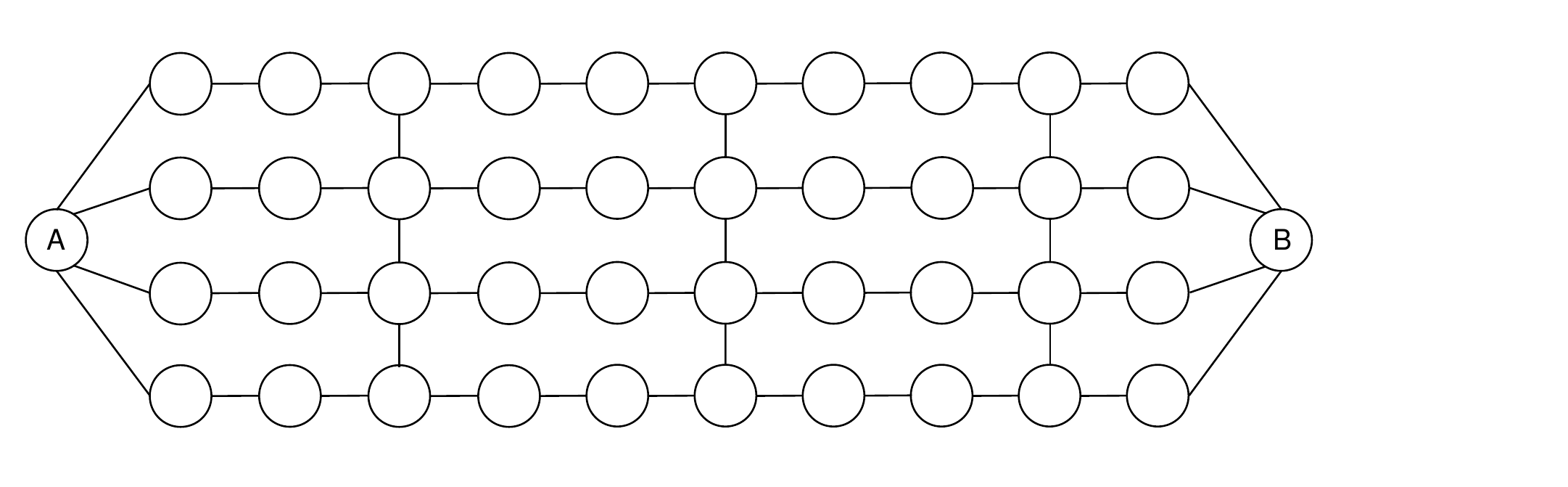}
    \caption{A network with interlinks placement in MOP 4-scheme. The 
configuration is an alternative to 5 disjoint paths (MNOP) scheme.}
    \label{fig:l-scheme}
\end{figure}

In the MNOPs scheme, the probability of hacking is $(1-(1-p)^n)^l 
\approx (n p)^{l}$. We will keep this approximation, but later, in 
Sec. \ref{sec:error}, we perform an analytical analysis and estimate 
its error. To calculate the probability of hacking the $l$-scheme 
we divide the sample probability space into events $H_k, \; 0 \leq k \leq nl$, 
where each event means that exactly $k$ of $nl$ nodes become corrupted. 
This is a complete and pairwise disjoint set of events. Therefore, 
according to the law of total probability and binomial distribution, 
probability of information leakage in the MOPs $l$-scheme is:
\begin{equation}
\label{eqn:gen-prob}
    P(H)=\sum_{k=0}^{n l} {{nl}\choose{k}} p^k (1-p)^{nl-k} P(H|H_k),
\end{equation}
where $H$ is an event related to hacking an $A-B$ vertex cut.
Of course, $P(H|H_k)=0$ if $k<l$. Therefore, we must calculate 
$P(H|H_k)$ for $k\geq l$. This is rather a difficult task. Instead, 
we perform an estimation that provides a useful upper bound on $P(H)$. 
Let us denote the number of minimal cuts for $l$-scheme of length $l$ by
\begin{equation}
c(l,n) :=P(H|H_l) {{n l}\choose{l}}.
\end{equation}
Obviously, $c(l,n)=n$ if $l=1$, and for $l \ge 2 $ we have the following 
theorem:
\begin{tw} \label{alpha}
Define:
\begin{equation}
    \alpha(l)=
\frac{2^{-l} \left(\left(1+l+\sqrt{l^2+6 l-7}\right)^l-\left(1+l-\sqrt{l^2+6 l-7}\right)^l \right) }{\sqrt{l^2+6 l-7}}, 
\label{alphalformula}
\end{equation} 
for $l \ge 2$  and $\alpha(1)=1$. Then
\begin{equation} \label{eqn:alpha_estimation}
 \alpha(l) \times(n-2l^2) \leq c(l,n) \leq \alpha(l)\times n.
\end{equation}
\end{tw}
The term $\alpha(l)$ gives the asymptotic average number of minimal cuts per column as $n \rightarrow \infty $. For real schemes, it serves as an upper bound.
\begin{proof}
 The vertex cut requires hacking at least one vertex in each row, and each vertex belongs to at least one cut (each column forms a cut), and if we settle on one vertex, we can easily find all cuts containing this vertex.
Let $g_{ij}$ be the settled vertex and 
\begin{equation} \label{cond}
1+l<j<n-l,
\end{equation}
we can formulate a procedure that determines all vertices of the next (or previous) row that can be used to form the cut:
\begin{itemize}
    \item  If degree $g_{ij}>2$, then the possible vertices are $g_{(i+1)(j+r)} ,\; r \in\{-l+1,-l+2,...,l-2,l-1\}$.
    \item  If degree $g_{ij}=2$, then the possible vertices are $g_{(i+1) m} ,\; m\in \{a...b\},\;$ where $a$ and $b$ are determined by conditions: $deg( g_{ia} )>2 , \; deg (g_{ib})>2,\;  b-a=l-1,\;  a < j < b$.
\end{itemize}

So, we can say that each vertex ``produces'' certain new vertices in each row.    
For vertices that do not meet \eqref{cond}, which can be intuitively described as ``boundary vertexes'', the procedure above must be slightly modified, but it will result that the number of ``produced'' vertices is smaller. Then, we calculate $\alpha(l)$ with the following recursion:

\begin{verbatim} 
f(n,dg2): %n-number of iteration, dg2-boolean variable true if degree >2
    if n==l return 1
    if dg2==True return 3*f(n+1,1)+2(l-2)*f(n+1,0) 
    if dg2==False return 2*f(n+1,1)+(l-2)*f(n+1,0)
alpha(l):
    return (f(l,1)+(l-2)*f(l,0))/(l-1)
\end{verbatim}

This can be alternatively rewritten as a pair of related sequences and expressed in the following form with matrix multiplication:
\begin{equation}
\alpha(l)=\frac{1}{l-1}
\begin{pmatrix}
    1 & l-2
   \end{pmatrix}
  \begin{pmatrix}
3 & 2(l-2) \\
2 & (l-2) 
\end{pmatrix}^{l-1}
\begin{pmatrix}
1 \\
1
\end{pmatrix}.
\label{AnalyticAlpha}
\end{equation}
The matrix in this formula can be diagonalized leading to Eq. 
\ref{alphalformula}. The expression in $\alpha(l)$ is overestimated 
because we neglected the ``boundary corrections''. Because the 
ranges of $r$ and $m$ are sometimes smaller (especially close 
to the sides of the graph), the estimation improves as $n$ 
increases. An underestimation is obtained if we do not count all 
vertex cuts that contain ``boundary vertices''.

\end{proof}
We now have analytical estimation for the first non-zero term in
\eqref{eqn:gen-prob}, where we neglect ``boundary corrections''. 
In fact, it can also be calculated numerically in time $O(nl)$ using 
dynamic programming. The algorithm is presented in \nameref{sec:appdynprog}.
Therefore, in the numerical calculation (for hypothetical testing of 
the algorithm), we will use the exact value of $\alpha(n,l)$, 
although for analysis it is convenient to use the upper bound as 
in Theorem \ref{alpha}. The difference is significant when $n$ has 
a similar order as $l$.

Calculating higher terms in \ref{eqn:gen-prob} pose a problem. 
However, it turns out that if $nlp < 1$ it can be approximated 
by geometric series. To do that, we first derive important observation:
\begin{lemat}
\label{beta}
 Consider 
 \begin{equation}
 \beta_{n,l}(k) :={{nl}\choose{k}} P(H|H_k),
 \end{equation}
 which is the number of cuts of size k in the MOPs l-scheme graph. The following inequality holds
 \begin{equation} \label{eqn:beta}
     \beta_{n,l}(k) \leq \beta_{n,l}(k-1)nl  \text{ for }  k-1 \geq l.
 \end{equation} 
 Therefore:
 \begin{equation}
 \label{beta_estimated}
      \beta_{n,l}(k) \leq \alpha(l) n (nl)^{k-l}.
 \end{equation}
 \end{lemat} 
 
The Lemma \ref{beta} is obtained by numerical, not analytical 
analysis. We do not posses a proper proof of its validity, but 
for practical use, we only need to confirm its correctness with 
given $n$ and $l$. This could be done numerically with brute 
force, generating all possible vertex sets, and checking if it 
is cut. In fact, due to the computational complexity of this method, 
it can only be performed on small graphs. We performed such 
calculations with all possible schemes with graph size < 30 
\emph{i.e} if $nl <30$. See more in \nameref{sec:betacalculation}.
An alternative method that could substantiate these results is its 
statistical testing \emph{i.e} generating a sample consisting of 
random sets of nodes and calculating the ratio of $A-B$ cuts, 
generated this way versus the sample size. An important observation 
that is worth stressing is that the problem of calculating the 
number of cuts in such a graph (and calculating probability of 
that there exists a path between $A$ and $B$) looks similar to 
the percolation problem. We did not follow this lead, but it 
could be useful if one would try to prove the lemma \ref{beta}.

With use of the above lemma we can now obtain:
\begin{tw} \label{estimate_lscheme_prob}
If $p$ is small enough, so it satisfies $n l p \leq r < 1$, then:
\begin{equation}
    P(H)\leq \frac{1}{1-r} \alpha(l)n  p^l.
\end{equation}
\end{tw} 
\begin{proof}
First, from \eqref{beta_estimated} we derive 
\begin{equation}
    P(H|H_k) = \frac{\beta_nl(k)}{{ nl\choose k}} \leq \frac{ \alpha(l)n(nl)^{(k-l)}} {{nl\choose k}}.
\end{equation}
Then \eqref{eqn:gen-prob} can be estimated using a geometric progression 
formula with initial term $\alpha (l) n p^l$ and common ratio $nlp$.
\end{proof}

It will be useful to know how much our algorithm is better than the 
MNOPs. Thus, for a given network, we introduce efficiency coefficient 
$\eta$, which approximate ratio of probability of hacking of our 
algorithm vs MMOP scheme:
\begin{equation}\label{eqn:eta-definition}
     \frac{P(H)}{(np)^{l+1}} \leq \frac{\frac{1}{1-r}\alpha(l)}{p n^{l}} :=\eta.
\end{equation}

As described before in Sec. \ref{sec:protocol}, we used the approximation 
$1-(1-p)^n \approx np$ for the classic scheme, which is not necessarily 
desired as is usually $1-(1-p)^n < np$. We refer to this concern in Sec.  \ref{sec:error}, in which we show that due to this fact we shall introduce 
numerical correction for $\eta$. However, for ``reasonable'' graphs the 
corrected efficiency is in the worst cases only about $1.7$ grater than 
the one predicted by Eq.\ref{eqn:eta-definition}. Therefore, it has a really 
small impact and we omit it.The sufficient condition for the $l$-scheme to 
perform better than the classic scheme is:
\begin{equation}
\label{eqn:working}
\eta  < 1,
\end{equation}

Consequently, the risk of hacking is at least $\eta$ times smaller 
(if both schemes use the same amount of links). Equation \eqref{eqn:working} 
allows us to study the performance of the algorithm and we see that for 
the constant $p$ we can obtain, asymptotically for very large graphs, 
very low values of $\eta$. However, as trust of nodes grows (so $p$ 
becomes smaller), the effectiveness of the algorithm decreases. In the 
same time we need to keep assumed condition $nlp<r<1$. This may leads 
to concerns about upper bounds on $p$. However, for ``realistic networks'' 
we can assume that $n<100, l<10, r<1/2$ ($\frac{1}{1-r}<2$) what, in 
worst case, gives $p<1/2000$, which is not very restrictive. In Sec. 
\ref{sec:numerical_analysis} we present numerical analysis for 
\eqref{eqn:working} which shows exactly in which ranges of $n$ and $p$ our 
algorithm is useful.

\subsection{Analysis of approximation in MNOPs scheme}
\label{sec:error}

Along chapters \ref{sec:attacks} and \ref{sec:multipathscheme} we use 
approximation $1-(1-p)^n \approx np$, which works for $np \ll 1$. Here, 
we analyse validity of this assumption and make correction to equation
\eqref{eqn:working}. First, let us note that just for classic scheme 
we shall demand that $np < 1$. If this does not hold (so $p>\frac{1}{n}$), 
then the probability of hacking single path is:
\begin{equation}
\label{eqn:szacunki}
    P=1-(1-p)^n > 1-\left(1-\frac{1}{n}\right)^n > 1-\frac{1}{e} \approx 0.63,
\end{equation}
which is unacceptable for any real network called ``secure''. 
This can also substantiate reality of assumption made in Theorem \ref{estimate_lscheme_prob} that $nlp < 1$. We define parameter $\mu$
\begin{equation}
\mu :=\frac{1}{n p} \geq 1, 
\end{equation}
which is a certain constant fixed for a network (path), this divide 
all networks models (paths) on certain classes, with fixed value of 
$\mu$. Analysing validity of approximation $1-(1-p)^n \approx np$ as
a function of $\mu$ gives us useful results, namely lower bound for 
probability of hacking. For a given class (characterized by $\mu$) consider:
\begin{equation}
    v_n(\mu)= \frac{1-(1-p)^n - np}{np} = \frac{1-(1-\frac{1}{\mu n})^n - \frac{1}{\mu }}{\frac{1}{\mu }} = \mu \left(1-\left(1-\frac{1}{\mu n}\right)^n\right)-1.
\end{equation}
Clearly $$ P= 1-(1-p)^n= (1+v_n(\mu)) np .$$
We have two following properties:
\begin{equation}
    v_n(\mu) > v_{n+1}(\mu),
\end{equation}  and 
\begin{equation}
  v(\mu):= \lim_{n \to \infty} v_n= -1 + \mu(1-e^{-1/\mu}).
\end{equation}
Let $\gamma(\mu) := 1+v(\mu )= \mu (1-e^{-1/\mu}) <1$.
Then 
\begin{equation} 
    1-(1-p)^n  \geq  \gamma(\mu) np.
\end{equation}
And therefore putting $\gamma(\mu) np$  instead of $np$ in \eqref{eqn:working}, we get more restrictive condition on $\eta$ 
(that refer to situation that classic scheme works better than assumed).
\begin{equation} \label{eqn:working-restrictive}
    \frac{P(H)}{(1-(1-p)^n )^{l+1}} \leq \frac{P(H)}{(np)^{l+1}} \frac{1}{\gamma(\mu)^{l+1}} \leq \eta \frac{1}{\gamma(\mu)^{l+1}} < 1.
\end{equation}
Function $\gamma(\mu)$ is presented on figure \ref{fig:gamma}.
\begin{figure}[h]
    \centering
    \includegraphics[scale=0.6]{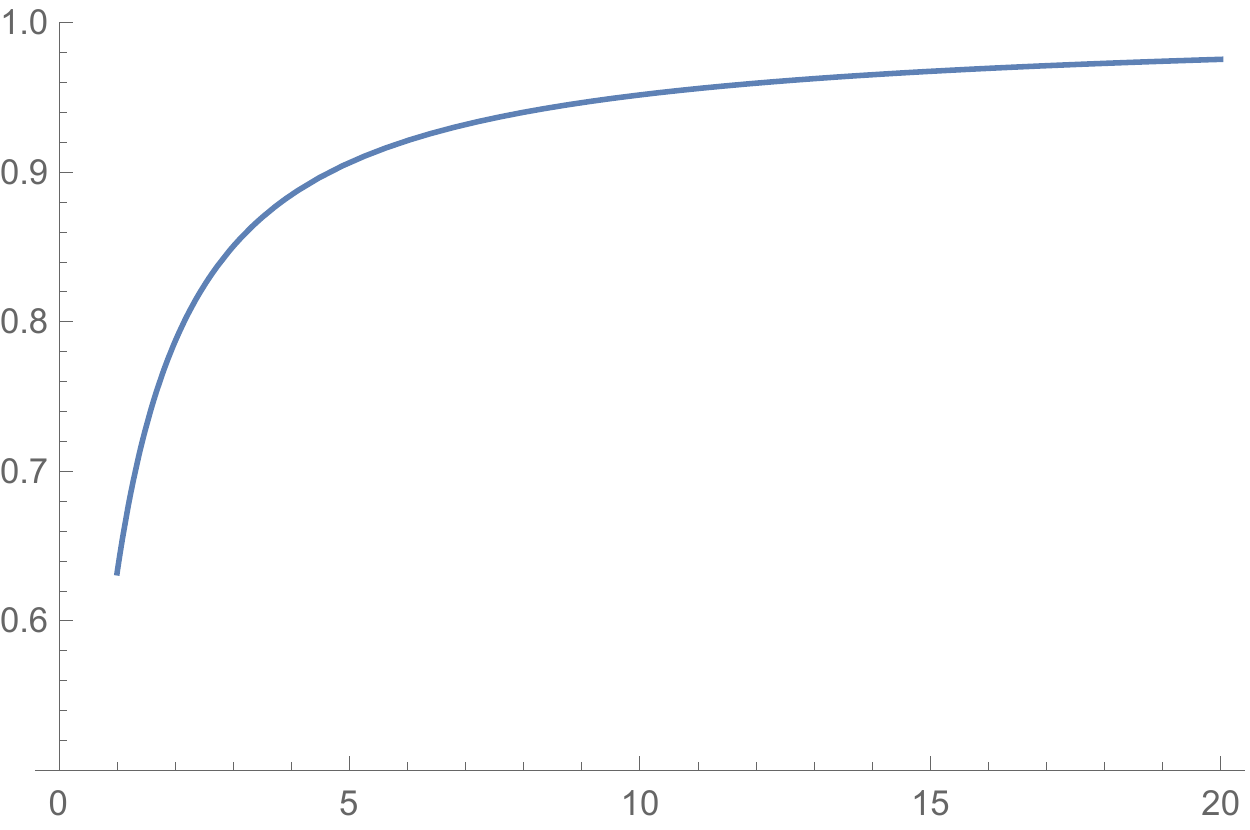}
    \caption{Plot of the function $\gamma(\mu)$.}
    \label{fig:gamma}
\end{figure}
We conclude that typically we shall deal with systems with parameter 
$\mu \ge 10$, this is done by posing arbitrarily but reasonable 
condition ($P<1/10$) and analysis as in equation \eqref{eqn:szacunki}. 
We assume that usually  $l\leq 10$, otherwise communication will be 
very expensive in resources. For such assumption we can estimate maximal
value of correction, namely $\frac{1}{\gamma(10)^{11}} \approx 1.72$.
This section does not change our conclusions significantly but is 
required for completeness of analysis. 

\subsection{Numerical analysis} \label{sec:numerical_analysis}

For given graph and parties communicating with $l$-scheme, parameters 
like $n$ and $l$ can be settled or easily estimated, but trust (equivalently probability of hacking) of intermediate nodes is very unclear to define 
and measure in reality. It can be connected with various events such as
random failure of network , corrupted labours, hackers or it could change 
in time. Therefore, it is rather impossible to declare its value on the 
stage of theoretical considerations. But still we expect that this trust 
(and respectively probability of hacking $p$) must be in reasonable range, 
for example $p\approx 0.1$ or $p\approx {10^{-27}}$ are certainly not, 
of course for security purpose $p$ should be as least as possible. 
Another problem is to find this range but we can analyse conditions imposed 
on $p$ steaming from algorithm structure of $l$-scheme and our considerations,
namely $nlp<r<1$ and $\eta <1$, in dependence of its parameters $n$, $l$, 
$r$. Those two conditions determine possible range of $p$ (first serves 
to establish upper bound and second for lower bound) in which MOPs 
$l$-scheme will work. Determined maximal and minimal values for $p$ 
in dependence for $l$ and $p$ with fixed parameter $r=1/10$ are presented 
in Figs. \ref{fig:minp}, \ref{fig:maxp} and \ref{fig:3D}.

Taking plans concerning the future quantum network spanned 
\emph{e.g.} across Europe, we suppose its size will refer to 
average number of intermediate nodes in one path $n<40$ 
(as expected size of QKD link is about 100 km) and maximal 
number of disjoint path $l< 6$. Therefore, applicability of MOPs 
$l$-scheme is under question and depend on real value of $p$. 

\begin{figure}[H]
    \centering
    \includegraphics[scale=0.7]{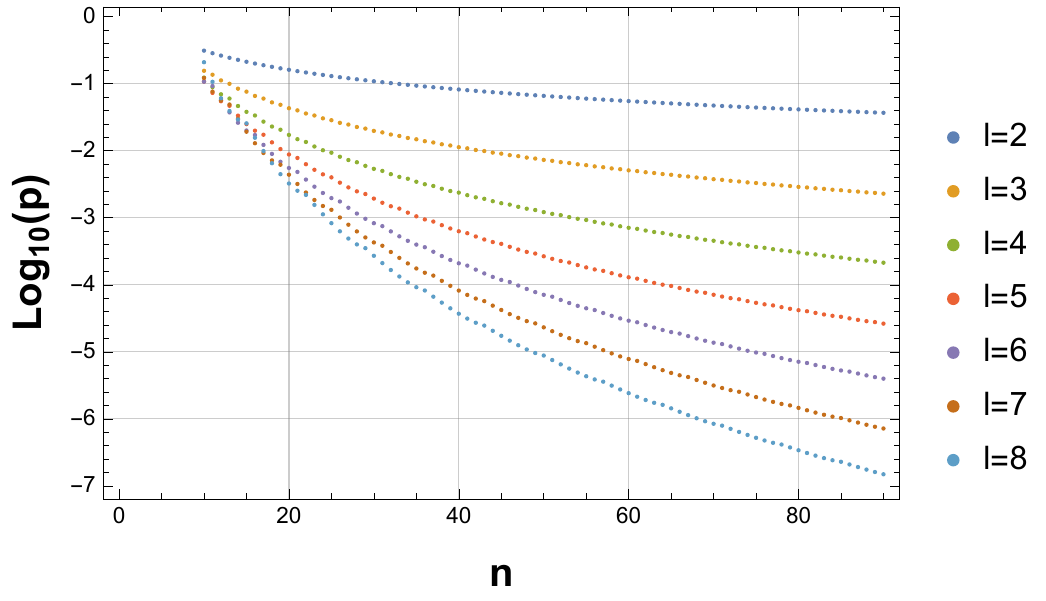}
    \includegraphics[scale=0.7]{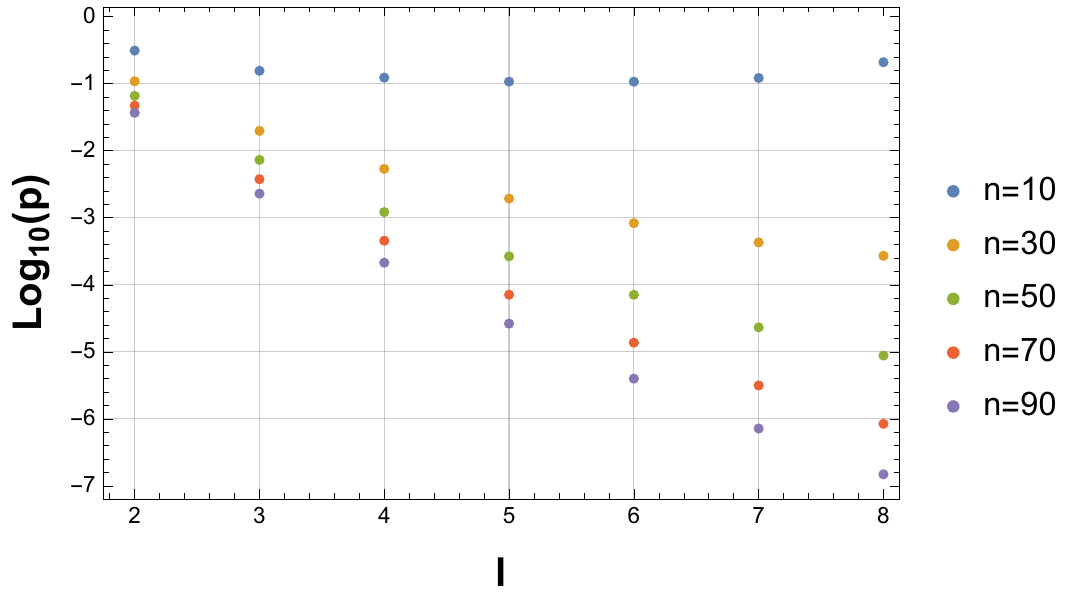}
    \caption{Minimum value of $p$ (expressed in logarithm of base 10) in dependence on $l$ and $n$ and $r=1/10$, satisfying condition $\eta<1$.}
    \label{fig:minp}
\end{figure}

\begin{figure}[H]
    \centering
    \includegraphics[scale=0.7]{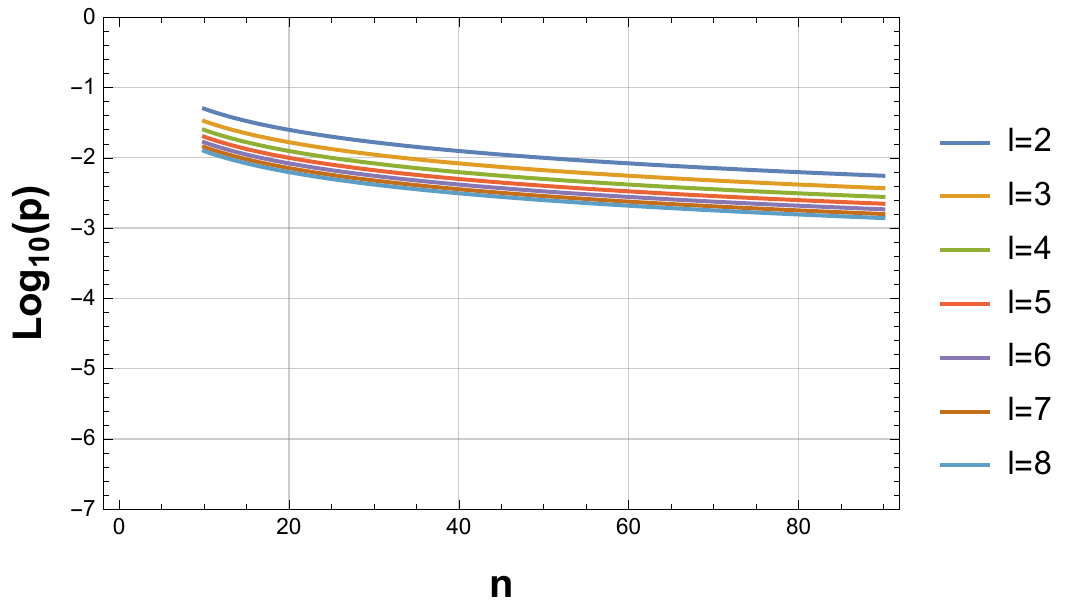}
    \includegraphics[scale=0.7]{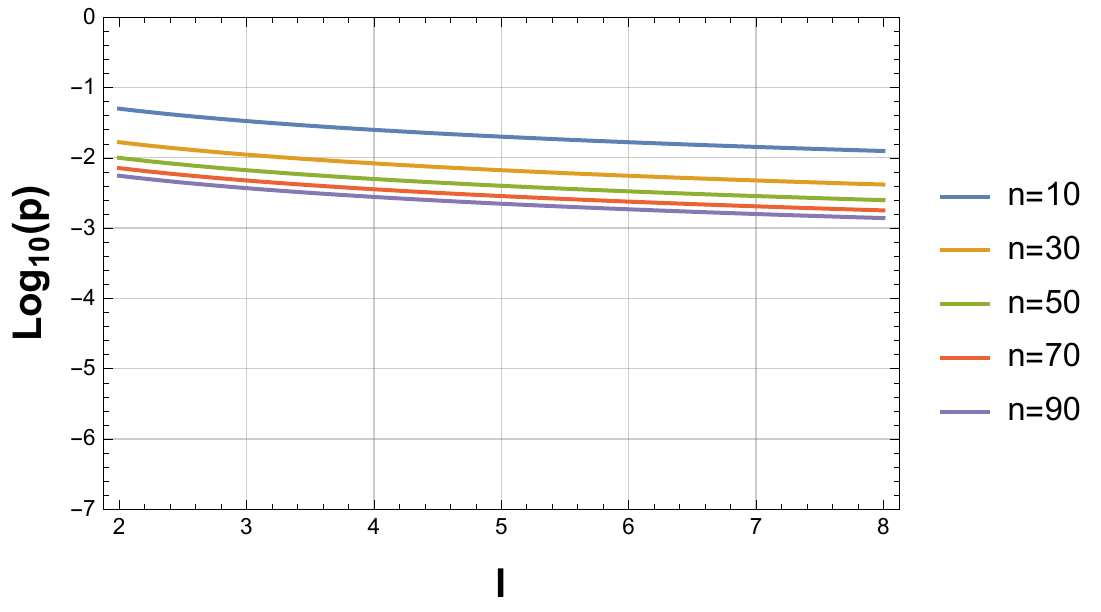}
    \caption{Maximal value of $p$ (expressed in logarithm of base 10) in dependence of $l$ and $n$ and , satisfying condition $nlp < r=1/10$.}
    \label{fig:maxp}
\end{figure}

\begin{figure}[H]
    \centering
    \includegraphics[scale=0.4]{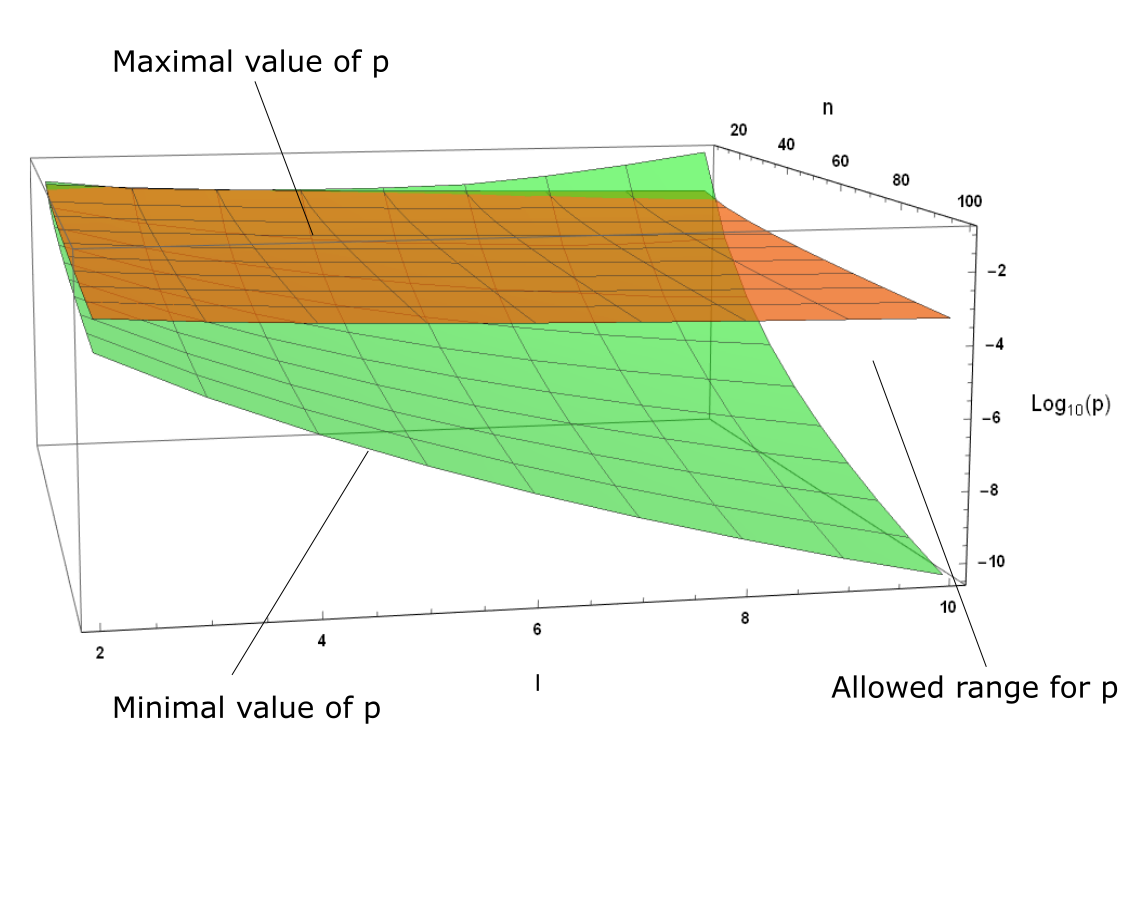}
    \caption{Range of $p$ which fulfill conditions $\eta <1$ and $nlp < 1/10 $, expressed with logarithm of $p$ in dependence of different $l$ and $n$.}
    \label{fig:3D}
\end{figure}

\section{Summary}

In this article, we have analyzed two scenarios of attack that can 
be performed on trusted nodes in the hybrid QKD network. First, 
describing the situation in which each trusted node could be 
compromised with a certain probability $p$ and second describing a 
correlated attack on a network with finite resources. For which 
case the risk of hacking is greater depends on individual parameters 
of network and attacking party, which are difficult to predict in 
reference to the real world. However, we can infer that with growing
network size, the second scenario is less less vulnerable to attacks. 

Next, we have described the scheme of communication in the QKD network
extending the multiple path scheme by the possibility of crossing
communicating paths - the MOPs scheme. This scheme uses the same amount 
of resources (QKD links) and can perform better under certain conditions, 
as analyzed in \ref{sec:numerical_analysis}. The graphic visualization 
of the most restrictive constraint (\emph{i.e.} on minimal value of $p$) 
is presented in Fig. \ref{fig:summary}.
\begin{figure}[H]
    \centering
    \includegraphics[scale=0.4]{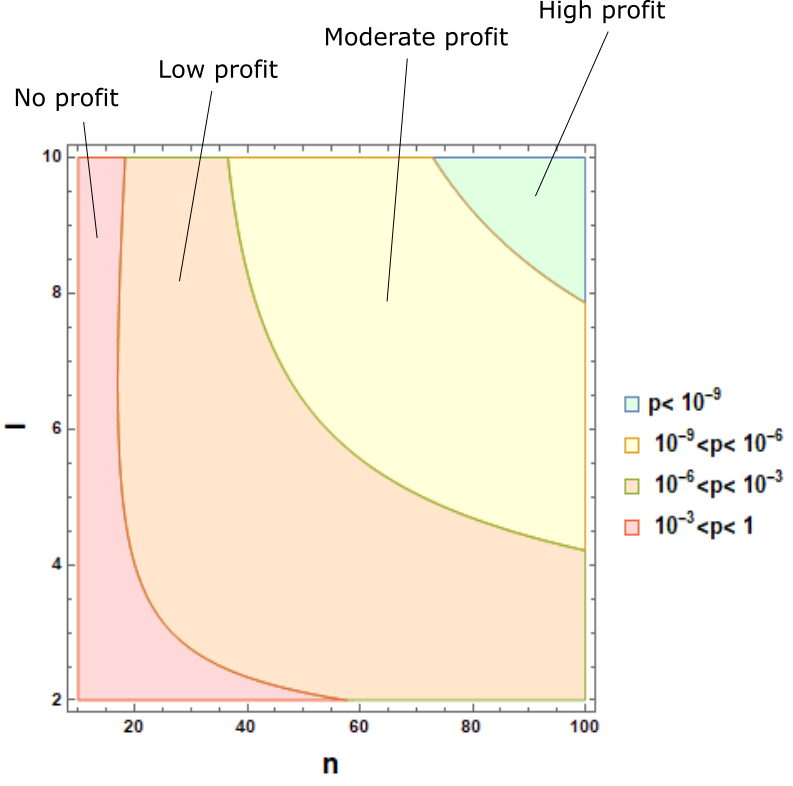}
    \caption{Regions of the allowed value of $p$ as a function of $l$ (number of paths) and $n$  (number of nodes in a path), under the condition $\gamma < 1$. The condition guarantees that the probability of hacking the MOP scheme is smaller than for the MNOP scheme. The colors represent degree of usefulness of the algorithm with green as the most useful and red as not useful.}
    \label{fig:summary}
\end{figure}

Unfortunately, these constraints are possible to fulfill for distant users 
in large-scale networks. 

The concept presented in Sec. \ref{sec:adding_single_link} is worth 
studying. It does not present a groundbreaking idea but introduces 
an interesting and inexpensive improvement to the QKD network. This 
work develops a new way of thinking about QKD multiple path algorithms 
in hybrid networks with trusted nodes. The proposed algorithm is not 
an optimal one. Considering different topology of interlinks could 
perform better but the simple model under investigation enabled 
analytical analysis. For this moment it is difficult to judge whether 
the presented concept will be useful in practice, yet it opens a new 
path for future development of QKD networks.

\section*{\label{sec:app_attack_example}Appendix 1}

\begin{figure}[h]
    \centering
    \includegraphics[scale=0.5]{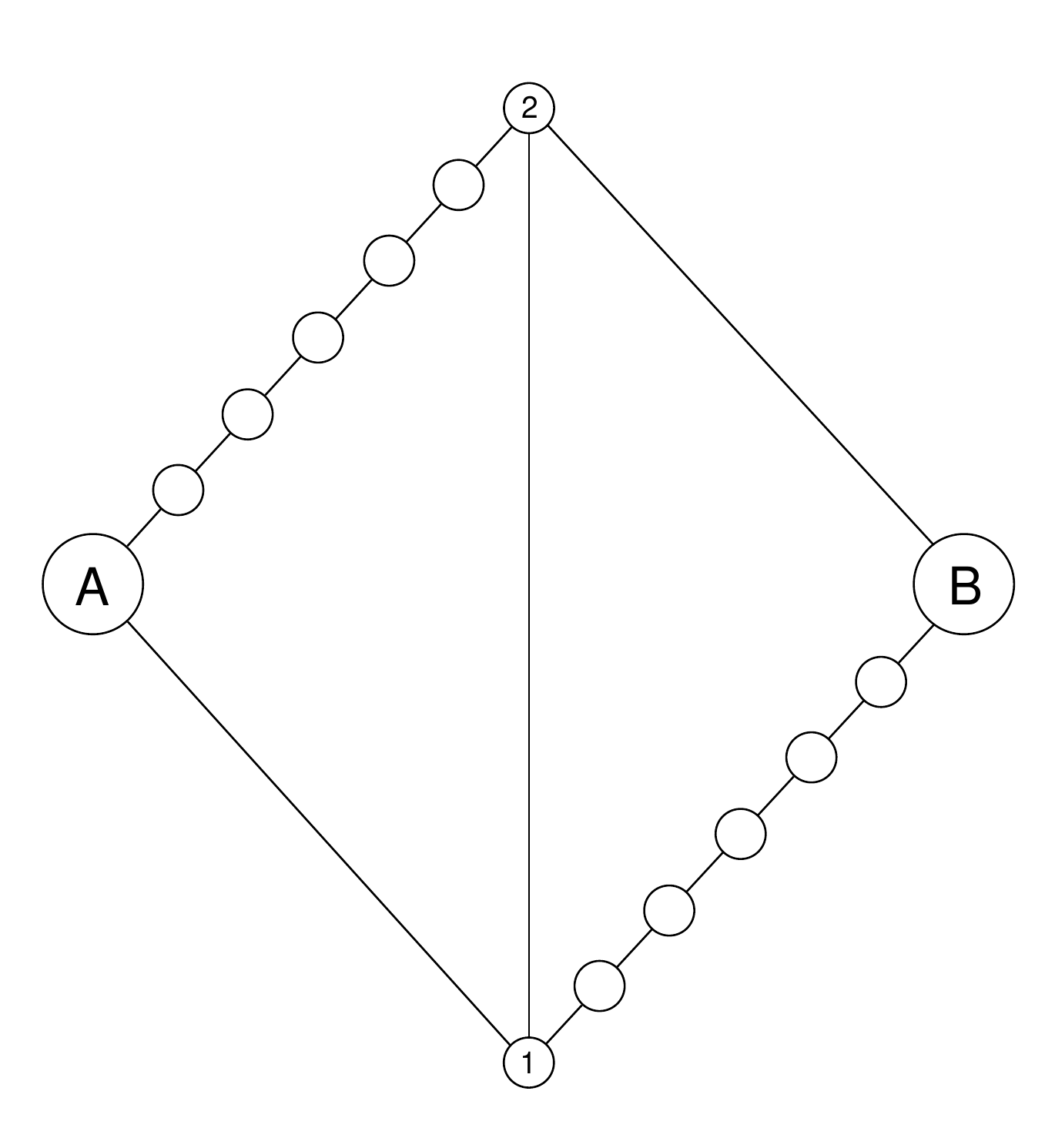}
    \caption{Example of network for which single path protocol is more optimal than multi-path one.}
    \label{fig:example}
\end{figure}

Here, we present an example where for the uncorrelated attack 
(even with uniform probability $p$ for each vertex), it is not 
always optimal to use as many paths as possible. We adopt the 
assumptions made in the discussion of this type of attack in Sec. 
\ref{uncorrelated-attack}. Consider the network presented in Fig. \ref{fig:example}, and let the number of intermediate nodes on path 
$A-2$ and $1-B$ (a path that goes from $A$ to $2$ and $1$ to $B$ 
but omits $1$ and $2$, respectively) be $n$, and the probability 
that any (intermediate) node becomes compromised is $p$. If we 
use one or two paths, then the probability of hacking is, respectively: 
 \begin{align}
&P_1= 2p, \\
&P_2= \left((n+1)p \right)^2 .
\end{align}
It is possible to satisfy the inequality $P_1<P_2$, which leads to 
the condition:
\begin{equation}\label{eqn:ap1}
\frac{2}{(n+1)^2} < p  .
\end{equation}
But for $n$ large enough, we can find $p$ small enough so 
Eq. \ref{eqn:ap1} do not contradict the condition assumed 
before, namely$ np \ll 1$. Therefore, even for the simplified 
version of the correlated attack model, it is not always 
optimal to use as many paths as possible.

\section*{Appendix 2}\label{sec:appNDMS_to_hopbyhop}

In the MNOPs scheme, each intermediate node has a connection 
to two other nodes and passes a secret key from one to another 
in a hop-by-hop fashion using a one-time pad. In the MOPs scheme, 
each node sends a classical message to Bob, which may be not 
efficient (since it generates a lot of traffic in the classic network), 
and we now show a certain alternative that reduces impact on key rate.

\begin{lemat}
\label{lem:alterhop-by-hop}
Alternatively to the hop-by-hop scheme, we can use the following 
procedure: 
In path $A-1-2...n-B$ vertex number 1 take the $XOR$ of keys $K_{A1}$ 
and $K_{12}$ and send it to the next node. If node $i \in{2...n}$ 
receives the message $M_{i-1}$ from the previous node, it takes $XOR$ 
of his shared keys and the message and sends it forward \emph{i.e.} 
$ M_i = M_{i-1} \bigoplus K_{i-1,i} \bigoplus K_{i,i+1} $. The difference 
is that instead of decoding and encoding operations, we do it in single step.
\end{lemat}

\begin{lemat}
Alternatively, to the MOPs method, we can use the following procedure. 
Having a graph, we find system of disjoint paths that cover all vertexes. 
This system will determine the next node for each vertex. While 
communicating, each node uses the procedure described in Lemma \ref{lem:alterhop-by-hop} with difference we additionally $XOR$ all 
interlinks key. Consequently, the message sent forward is: 
$M_i = M_{i-1} \prod_{(j,i)} \bigoplus K_{i,j} $, where $(j,i)$ 
denotes adjointness of the nodes $j$ and $i$.  
\end{lemat}

One can check that this extension of the MOPs scheme omits unwanted 
procedure in which each node sends message to Bob and at the same time 
keeps benefits of utilization of interlinks. A useful example is presented 
in Fig. \ref{fig:NDMShop}.
\begin{figure}[H]
   \hspace{-1cm}
    \includegraphics[scale=1]{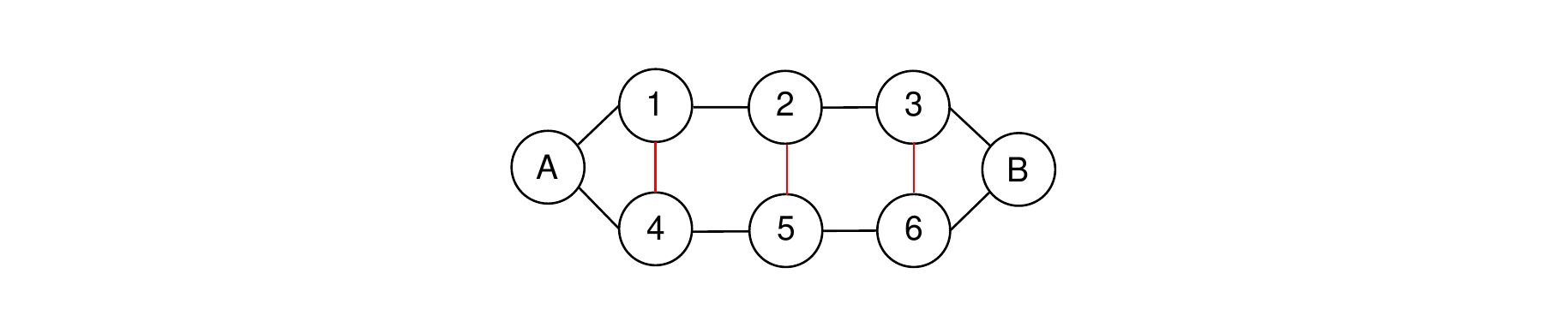}
    \caption{Example of graph for classic and NDMS scheme, red links represent additional interlinks in NDMS scheme. In NDMS following messages will be created:
    $M_{12}=K_{A1} \bigoplus K_{12} \bigoplus K_{14}$,
    $M_{23}=M_{12} \bigoplus K_{23} \bigoplus K_{25} \bigoplus K_{12} $,
    $M_{3B}=M_{23} \bigoplus K_{3B} \bigoplus K_{36} \bigoplus K_{23} $,
    $M_{45}=K_{A4} \bigoplus K_{45} \bigoplus K_{14} $,
    $M_{56}=M_{45} \bigoplus K_{56} \bigoplus K_{25} \bigoplus K_{45} $,
    $M_{6B}=M_{56} \bigoplus K_{6B} \bigoplus K_{36} \bigoplus K_{56}$.
    }
    \label{fig:NDMShop}
\end{figure}

\section*{Appendix 3}
\label{sec:appdynprog}

Here, we present a dynamical programming algorithm to calculate 
the exact value of $c(l,n)$, the greatest discrepancy in the result 
shows up when $n$ is close to $l$:

\begin{verbatim} 
algorithm alpha is
    input: n and l parameters of graph
    output: number of minimal cuts separating A and B

   create two dimensional array tab[l,n]
    for j in 1.. n : 
        tab[1,j]=1 

    for each row i in 2...n:
        for each cell in given row tab[i,j] , j in 1,n :
            calculate range of possible k (a,b) such that nodes(i,j) and (i-1,k) 
            can be in vertex cut for k in a...b :
            tab[i,j]+=tab[i-1,k]

    return sum of tab[l,i] 
\end{verbatim}

\section*{Appendix 4}
\label{sec:betacalculation}

Here, we present the numerical calculation performed to substantiate 
the lemma \ref{beta} from Sec. \ref{subsec:formalconsiderations}. 
For each graph, each subset of size k can be generated and checked 
to see whether it separates $A$ and $B$ (if $A$ and $B$ are in 
disjoint components). The numerical values of the function $\beta_{n,l}(k)$
for small graphs are presented in Fig. \ref{fig:matrices}.
\begin{figure}[H]
\begin{center}
    \centering
    \includegraphics[scale=0.6]{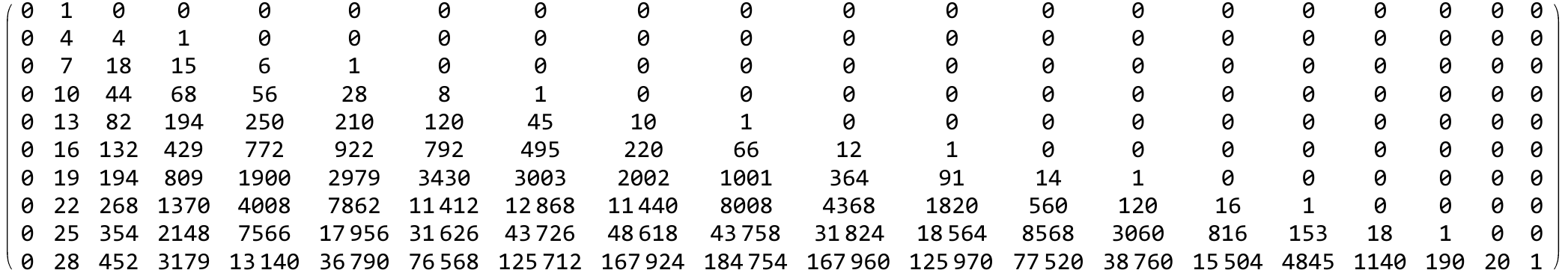}    \vspace{0.3cm}
    
    \includegraphics[scale=0.6]{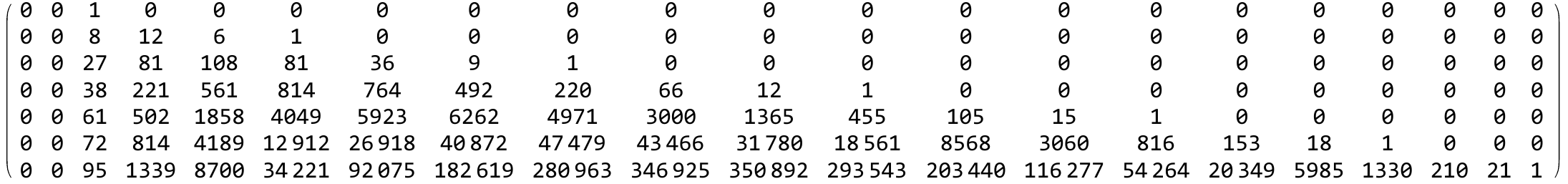}    \vspace{0.3cm}
    
    \includegraphics[scale=0.6]{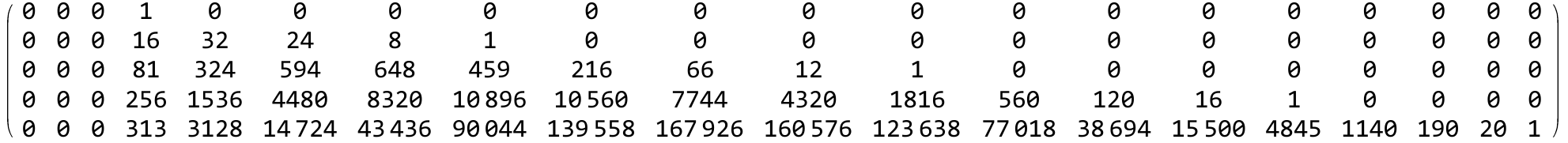}    \vspace{0.3cm}
     
    \includegraphics[scale=0.6]{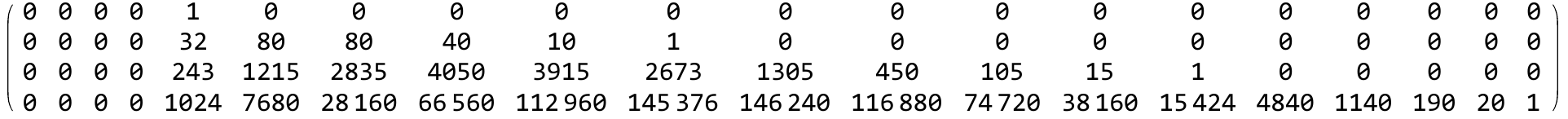}    \vspace{0.3cm}
    
    \caption{Values of function $\beta_{n,l}(k)$. Column number 
    (starting from 1) depict value of $k$ , and row number value 
    of $n$, and $l$ is constant within given matrix respectively 
    $l=2, 3, 4, 5$. Through computation time is exponential in the 
    size of the system (which equals $nl$), calculations were made 
    for graphs of size maximal size about 30.}
    \label{fig:matrices}
\end{center}
\end{figure}

From this analysis, one can deduce that the ratio
$\frac{\beta_{n,l}(k+1)}{\beta_{n,l}(k)}$ is maximal for $k=l$.
We also derive numerical observations suggesting that inequality
$\eqref{eqn:beta}$ ``cannot be strengthened'', by which we mean 
that:
\begin{equation}
 \frac{\frac{\beta_{nl}(l+1)}{\beta_{nl}(l)}}{nl} \rightarrow 1, 
\end{equation}
as $n$ grows. In Fig \ref{fig:limit} the numerical result obtained for 
$l=2$ are shown. For $l>3$, the above thesis seems to hold but collecting 
data for many points is too time consuming and, therefore, we do not 
present more results. 
\begin{figure}[H]
    \centering
    \includegraphics[scale=0.8]{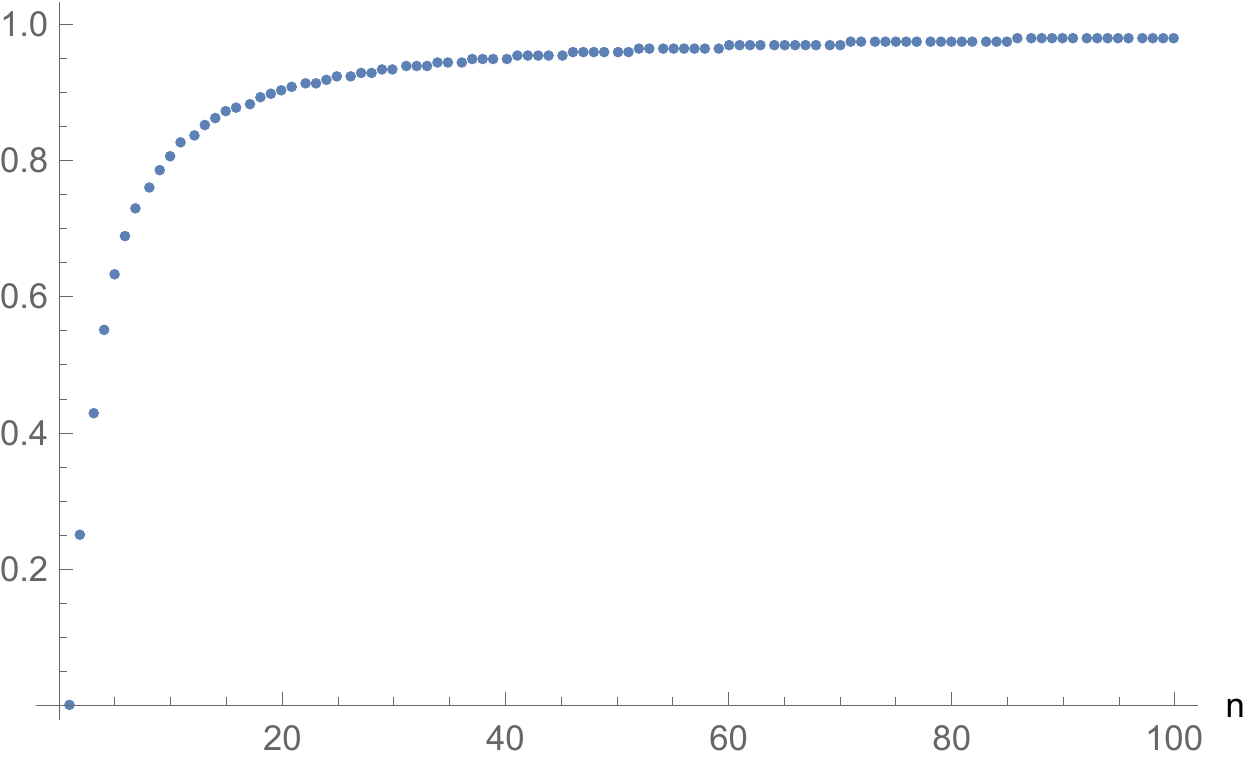}
    \caption{Function $\frac{\frac{\beta_{nl}(l+1)}{\beta_{nl}(l)}}{nl}$ with $l=2$ in dependence of $n$.}
    \label{fig:limit}
\end{figure}

\bibliographystyle{unsrt}
\bibliography{Bibliography.bib}


\end{document}